\documentclass[twocolumn]{autart}    

\usepackage{graphicx}          
\usepackage{amssymb, amsmath,amsfonts,verbatim, mathtools,enumitem,dsfont}
\usepackage{comment}
\usepackage{diagbox}
\usepackage{times,algpseudocode, algorithm}
\usepackage[style=base,justification=centering]{subcaption}
\usepackage{hyperref}

\allowdisplaybreaks[2]
\newcommand{\ones}{\mathbf 1}
\newcommand{\reals}{{\mathbb{R}}}

\newcommand{\argmin}{\mathop{\rm argmin}}
\newcommand{\argmax}{\mathop{\rm argmax}}

\newcommand{\norm}[1]{\left\lVert#1\right\rVert}
\newcommand{\mnorm}[1]{{\left\vert\kern-0.25ex\left\vert\kern-0.25ex\left\vert #1 
    \right\vert\kern-0.25ex\right\vert\kern-0.25ex\right\vert}}

\newcommand{\mc}{\mathcal}

\newtheorem{proposition}{Proposition}
\newtheorem{assumption}{Assumption}
\newtheorem{problem}{Problem}

\algnewcommand\algorithmicforeach{\textbf{for each}}
\algdef{S}[FOR]{ForEach}[1]{\algorithmicforeach\ #1\ \algorithmicdo}

\newcommand{\thickbar}[1]{\mathbf{\bar{\text{$#1$}}}}
\newcommand{\thickhat}[1]{\mathbf{\hat{\text{$#1$}}}}
\usepackage{xcolor}
\newcommand{\tc}[2]{\textcolor{#1}{#2}}
\newcommand{\miny}{\min_{y\in\mc{Y}(P,p)}}
\newcommand{\argminy}{\argmin_{y\in\mc{Y}(P,p)}}
\newcommand{\change}[1]{\tc{black}{#1}}
\newcommand{\ctwo}[1]{\tc{black}{#1}}

\newcommand{\arxiv}[1]{#1} 

\begin{document}

\begin{frontmatter}
\title{Adaptive Constraint Satisfaction for Markov Decision Process Congestion Games: Application to Transportation Networks}
\author[AA]{Sarah H.Q. Li},  
\author[AA]{Yue Yu},  
\author[NM]{Nico Miguel}, 
\author[EE]{Daniel Calderone}
\author[EE]{Lillian J. Ratliff}
\author[AA]{Beh\c cet A\c c\i kme\c se}

\address[AA]{Department of Aeronautics and Astronautics, University of Washington, Seattle, USA. (e-mail:\{sarahli, yueyu, behcet\}@uw.edu).}
\address[EE]{Department of Electrical Engineering, University of Washington, Seattle, USA.
 (e-mail: \{djcal, ratliffl\}@uw.edu)}
\address[NM]{Department of Aeronautics and Astronautics, Purdue University, USA.  (e-mail: nmiguel@purdue.edu)}

\begin{abstract}                
Under the \emph{Markov decision process (MDP) congestion game} framework, we study the problem of enforcing population distribution constraints on a population of players with stochastic dynamics and coupled congestion costs. Existing research demonstrates that the constraints on the players' population distribution can be satisfied by enforcing tolls. However, computing the minimum toll value for constraint satisfaction requires accurate modeling of the player's congestion costs. 
Motivated by settings where an accurate congestion cost model is unavailable (e.g. transportation networks), we consider a MDP congestion game with \emph{unknown} congestion costs.
We assume that a constraint-enforcing authority can repeatedly enforce tolls on a population of players who converges to an $\epsilon$-optimal population distribution for any given toll. 
We then construct a myopic update algorithm to compute the minimum toll value while ensuring that the constraints are satisfied on average.
We analyze how the players' sub-optimal responses to tolls impact the rates of convergence towards the minimum toll value and constraint satisfaction.
Finally, we construct a congestion game model for Uber drivers in Manhattan, New York City (NYC) using data from the Taxi and Limousine Commission (TLC) to illustrate how to efficiently reduce congestion while minimizing the impact on driver earnings.
\end{abstract}


\end{frontmatter}

\section{Introduction}\label{sec:intro}
Congestion games play a fundamental role in engineering~\cite{rosenthal1973class}. In large-scale networks such as urban traffic and electricity markets, congestion games capture how the competition among self-motivated decision-makers, known as the \emph{players}, impacts network-level trends~\cite{patriksson2015traffic,etesami2020smart}. In particular, when all the players are equipped with identical congestion costs and transition probabilities in a finite state-action space, the game can be modeled as a \emph{Markov decision process (MDP) congestion game}~\cite{calderone2017markov}. 

We study the feasibility of using tolls to enforce system-level constraints on a game in which both the players and the system operator do not know the true congestion costs. This is motivated by a plethora of constrained large-scale network systems such as autonomous swarms~\cite{demir2015decentralized}, urban transportation~\cite{lin2018environmental} and competitive electricity markets~\cite{miguelez2004practical}. For example, the Department of Transportation could meet carbon-emission targets by tolling fossil fuel vehicles on freeways~\cite{lin2018environmental}. In electricity markets, power auctions are followed by a procedure to predict and eliminate power violations through initializing offline generators~\cite{miguelez2004practical}. To minimize the incurred initialization cost, the system operator can use tolls during the auction to limit voltage demands.
Since tolling is a common and easily implementable mechanism in networked systems~\cite{surge_pricing},  we assume the system operator can freely impose tolls on the players. 

The \emph{minimum toll value} for satisfying constraints can be computed as a function of the congestion costs~\cite{us}. However, extracting the congestion cost is difficult when the player objectives are complex and unknown. This is also true in simulation engines and higher complexity models, where the effect of a given toll can be computed, but not the minimum toll value itself.
As such, we assume that an oracle exists who can compute the $\epsilon$-optimal solution for a game with a known toll. 
The inexact oracle is motivated by model-free learning algorithms that can approximate the Nash equilibria for routing games with unknown link costs~\cite{zhou2020reinforcement,krichene2015}. 

\textbf{Contributions.}
We derive a gradient-based tolling algorithm that enforces linear population distribution constraints on a class of MDP congestion games with unknown but strictly increasing congestion costs. The algorithm requires access to an inexact oracle that takes an input toll and returns an $\epsilon$-optimal population distribution for the toll-augmented game. We show a direct relationship between the $\epsilon$-optimal population distribution and an inexact gradient of the toll-augmented game with respect to the toll. We bound the following quantities as functions of the oracle's sub-optimality $\epsilon$: convergence of  1) the average toll value towards the minimum toll value, 2) the average population distribution towards the optimal distribution under the minimum toll, 3) the average constraint violation towards zero. Finally, we construct an MDP congestion game model using real-world data from NYC TLC~\cite{new_york_taxi} to demonstrate how our algorithm reduces Uber driver congestion in Manhattan.

The rest of this paper is organized as follows. In Section~\ref{sec:litReview}, we review related works. In Section~\ref{sec:prelim}, we  \ctwo{outline} MDP congestion games. In Section~\ref{sec:constraining_game}, we  \ctwo{formulate} the toll-augmented game and the inexact oracle. In Section~\ref{sec:dual_ascent_algorithm}, we  \ctwo{present} the tolling algorithm and prove its convergence properties. In Section~\ref{sec:ny_rideshare},  \ctwo{we create a queuing network MDP to reduce ride-share congestion levels in Manhattan, NYC. }
\section{Literature Overview} \label{sec:litReview}
MDP congestion games~\cite{calderone2017markov} are related to non-atomic routing games~\cite{beckmann1952continuous,patriksson2015traffic}, stochastic games~\cite{shapley1953stochastic}, and mean field games~\cite{lasry2007mean}, but differ in modeling assumptions. 
MDP congestion games extend non-atomic routing games by generalizing the player dynamics from deterministic to stochastic. 
Stochastic games assume that the player costs differ and are functions of the joint policy. MDP congestion games assume that the player costs are identical and are functions of the population distribution~\cite{yu2019primal}.
Finally, MDP congestion games are analogous to a discrete mean field game where the continuous stochastic processes are discretized in time, state, and action spaces.
We show that these differences in assumptions enable MDP congestion games to more easily model large-scale networks such as transportation systems.

Tolling schemes for non-atomic routing games have been studied under capacitated traffic assignment literature~\cite[Sec. 2.8.2]{patriksson2015traffic}. Adaptive game incentive design has also been considered in deterministic and stochastic settings in~\cite{ratliff2018adaptive} for players without MDP dynamics. Presently, we adopt a form of adaptive incentive design \ctwo{that guarantees} \emph{constraint satisfaction}. Tolling to satisfy external objectives is more generally interpreted as a Stackelberg game between a leader and its followers~\cite{swamy2007effectiveness}. Techniques for updating the Stackelberg leader's actions to optimize the social cost of its followers are derived in~\cite{roth2016watch}. Tolling non-atomic games \ctwo{under MDP dynamics with unknown congestion costs is the topic of this paper}.

Our minimum toll computation algorithm is an inexact gradient descent~\cite{devolder2014first}, and has been applied to settings such as distributed optimization and model predictive control~\cite{fazlyab2018distributed,necoara2013rate}. 
In game theory, the inexact gradient descent method has been applied to computing the Nash equilibria of a two player min-max game~\cite{ostrovskii2021efficient}. However, it has yet to be applied to constraint satisfaction with approximate Wardrop equilibria. 
\section{MDP Congestion Game} \label{sec:prelim}
\noindent\textbf{Notation.} The notation $[K]=\{0, \ldots, K-1\}$ denotes an index set of length $K$, $\reals (\reals_+)$ denotes a set of real (non-negative) numbers, $\ones_N$ denotes a vector of ones of size $N$, and $[x]_+ = \max\{x, 0\}$ denotes a vector-valued function in which $\max$ is element-wise applied to vectors $x$ and $0$.

Consider a continuous population of players, each with identical MDP dynamics and congestion costs over a state-action set $[S]\times[A]$ for $(T+1)$ time steps. 
\ctwo{Under the non-atomic game assumption, an individual player's probability distribution belongs to the zero measure subset of the population distribution (see~\cite[Sec.2]{krichene2015} for details).}
Presently, we only deal with the population distribution. We denote the set of feasible population distributions as $\mc{Y}(P,p)$, given by
\begin{equation}
    \begin{aligned}\label{eqn:feasibleY}
    \textstyle\mc{Y}(P,p)  = \{y \in \reals_+^{(T+1)S A}\Big| \, &
    \textstyle{\sum}_a y_{t+1,sa} =  {\sum}_{s',a}P_{tss'a}y_{ts'a}, \\
    & \textstyle{\sum}_ay_{0sa} = p_s \},
    \end{aligned}
\end{equation}
where $y_{tsa}$ is the portion of the playing population who takes action $a$ from state $s$ at time $t$. We emphasize that $y$ is a vector in $\reals_+^{(T+1)SA}$ whose coordinates are ordered as 
\begin{multline}\label{eqn:y_vector_order}
     y =  \left[ \begin{matrix}
            y_{000}&\ldots&  y_{010}&\ldots& y_{100}&\ldots &y_{T(S-1)(A-1)}\end{matrix}\right]^\top
\end{multline}
The transition dynamics are given by $P \in \reals^{T\times S\times S \times A}$, where $P_{tss'a}$ denotes the transition probability from state $s'$ to $s$ under action $a$ at time $t$. The transition dynamics satisfy $\sum_{s' \in [S]}P_{ts'sa} = 1$, and $P_{ts'sa}\geq 0,\ \forall\ (t, s', s, a) \in [T]\times [S]\times [S]\times [A]$.
The initial population distribution is given by $p \in \reals_+^{S}$, where $p_s$ denotes the portion of the \ctwo{playing population} in state $s$ at $t = 0$.

At time $t$, each player incurs a cost as a function of $y$, $\ell_{tsa}:\reals^{(T+1)SA} \rightarrow \reals$. We collect $\ell_{tsa}$ into a cost vector $\ell(y) \in \reals^{(T+1)SA}$ under the same ordering as $y$ in~\eqref{eqn:y_vector_order}. 
Similar to MDP literature, a player's expected cost-to-go at $(t, s, a)$ is its $Q$-value function, recursively defined as 
\begin{equation}\label{eqn:q_value_functions}
    Q_{tsa}(y) = \begin{cases}
    \ell_{tsa}(y) & t = T\\
    \ell_{tsa}(y) + \textstyle\underset{s'}{\sum}P_{ts'sa}\underset{a' \in [A]}{\min}Q_{t+1, s'a'}(y) & t \in[T]
    \end{cases}
\end{equation}
In an MDP congestion game, the $Q$-value functions depend on the players' collective action choices through $y$, the population distribution.
Each player minimizes its own Q-value function by choosing individual actions.
An \emph{MDP Wardrop equilibrium} $y^\star$ is reached if no player can unilaterally decrease its Q-value function further by changing its actions.
\begin{defn}[MDP Wardrop Equilibrium \cite{calderone2017markov}]\label{def:wardrop}
A population distribution $y^\star$ is a MDP Wardrop equilibrium if for every $(t,s,a) \in [T+1]\times [S]\times [A]$,
\begin{align}\label{eqn:wardropDef}
   y^\star_{tsa}> 0 \Rightarrow Q_{tsa}(y^\star) \leq Q_{tsa'}(y^\star), \quad \forall \ a' \in [A]
\end{align}
The set of $y^\star \in \mc{Y}(P,p)$ that satisfies~\eqref{eqn:wardropDef} is denoted by  $\mc{W}(\ell)$.
\end{defn}
At $y^\star$~\eqref{eqn:wardropDef}, all positive portions of the playing population distribution exclusively take actions with the lowest Q-values.
\begin{rem}
\ctwo{
MDP congestion games and stochastic games are multi-player extensions of the MDP via its linear program formulation.  
The coupling quantities between players, the \emph{population distribution} in MDP congestion games and the \emph{joint policy} in stochastic games~\cite{shapley1953stochastic}, are the primal and dual variables of the MDP~\cite[Eqn 6.9.2]{puterman2014markov}, respectively.}
\end{rem}
If $\ell$ is a continuous vector-valued function and there exists an explicit potential function $F$ satisfying $\nabla F(y) = \ell(y)$, then the MDP congestion game is a \emph{potential game}~\cite{monderer1996potential}.
\begin{proposition}\cite[Thm.1.3]{calderone2017markov} \label{prop:potential_game} \ctwo{Given} the MDP congestion game cost vector $\ell$, if a potential function $F$ satisfies
\begin{equation}
  \nabla F(y) = \ell(y), \quad F:\reals^{(T+1)\times[S]\times[A]}\mapsto \reals,\label{eqn:potential_def}  
\end{equation}
then the MDP Wardrop equilibrium is given by the optimal solution of
\begin{equation}\label{eqn:mdpcg_general}
\underset{y}{\min}\,F(y), \, \, \text{s.t. } y \in \mc{Y}(P,p).
\end{equation}
\end{proposition}
Games of form~\eqref{eqn:mdpcg_general} can be solved by convex optimization techniques~\cite{us}.
Using an MDP congestion game's potential function, we can characterize the degree of sub-optimality for any feasible population distribution within $\mc{Y}(P, p)$.
\begin{defn}[$\epsilon$-MDP Wardrop equilibrium]\label{def:epsilon_wardrop} For a game with the cost vector $\ell$,
potential function $F$~\eqref{eqn:potential_def}, MDP Wardrop equilibrium $y^\star$~\eqref{eqn:wardropDef}, \ctwo{ 
 and $\epsilon > 0$, the set of $\epsilon$-MDP Wardrop equilibria is given by}
 \begin{equation}\label{eqn:approx_WE}
\mc{W}(\ell, \epsilon) := \{ \thickhat{y}(\epsilon) \in \mc{Y}(P, p) \ | \ F\big(\thickhat{y}(\epsilon)\big) \leq F(y^\star)  + \epsilon\}.     
 \end{equation}
\end{defn}
Among cost vectors $\ell$ that have explicit potential functions, we focus on those that are strongly convex~\cite[Eqn B.6]{bertsekas1999nonlinear}.
\begin{assumption}\label{ass:strong_convexity}
The cost \ctwo{vector} $\ell$ has an explicit potential $F$~\eqref{eqn:potential_def} that is $\alpha$-strongly convex for all $y \in \mc{Y}(P, p)$.
    \[\textstyle\nabla_y \ell(y) \succeq \alpha I_{M\times M}\in \reals^{M\times M},\ M = (T+1)SA, \ \alpha > 0. \]
\end{assumption}
Assumption~\ref{ass:strong_convexity} implies congestion in \emph{all} state-action costs. To model games in which \emph{some} state-action costs are constant, we can approximate the constant costs by increasing functions with infinitesimal growth rates.

\begin{rem}
If  $\ell_{tsa}:\reals_+ \mapsto \reals \ \forall \ (t,s,a) \in [T+1]\times[S]\times[A]$ are scalar functions, then Assumption~\ref{ass:strong_convexity} implies that each $\ell_{tsa}$ strictly increases and satisfies $\alpha | y_{tsa} - y_{tsa}'|\leq |\ell_{tsa}(y_{tsa}) - \ell_{tsa}(y_{tsa}')|$. Its potential is also given by 
\begin{equation}\label{eqn:integral_potential}
   \textstyle F_0(y) ={\sum_{t,s,a}} \int_{0}^{y_{tsa}}\ell_{tsa}(u) d u.
\end{equation}
\end{rem}
For an $\epsilon$-MDP Wardrop equilibria $\thickhat{y}(\epsilon)$, 
Assumption~\ref{ass:strong_convexity} implies $\norm{\thickhat{y}(\epsilon) - y^\star}^2_2 \leq \frac{2\epsilon}{\alpha}$.

\section{Tolling for constraint satisfaction}\label{sec:constraining_game}
In this section, we formulate system-level constraints using affine population distributions inequalities and relate the inexact oracle of the tolled MDP congestion game to an $\epsilon$-MDP Wardrop equilibrium.
Affine constraints cover many design requirements for large-scale networks.
\change{As \ctwo{discussed in the introduction, meeting carbon emission goals in transportation and minimizing generator initialization costs in power grids are affine constraints on the fossil fuel vehicle population and local grid voltages, respectively.}}
\begin{defn}[Affine Constraints]\label{def:affineConstraints} The set of population distribution constraints is given by
\begin{equation}\label{eqn:constraint_def}
   \mc{C} = \textstyle\big\{y \in \reals_+^{(T+1)SA} \ | \ Ay - b \leq 0\big\} 
\end{equation}
where $A \in \reals^{C\times(T+1)SA}$, $b\in \reals^C$, and $0 \leq C < \infty$ denotes the total number of constraints imposed.
\end{defn}


Let $A_i \in \reals^{(T+1)SA}$ be the $i^{th}$ row of $A$. Instead of searching over all possible tolls, we only consider tolls of the form $\tau_i A_i\in \reals^{(T+1)SA}$ for $\tau_i \in \reals_+$. This formulation ensures that $\tau_i$ only affects the $(t,s,a)$ component of $\ell$ when $A_{i,tsa}$ is non-zero, where the \change{toll} magnitude is controlled by $\tau_i$.
We denote the toll-augmented game cost \ctwo{vector} as
\begin{equation}\label{eqn:augmented_cost}
    \ell_\tau(y) := \ell(y)  + A^\top \tau, \quad \tau \in \reals_+^C.
\end{equation}
When $\ell$ satisfies Assumption~\ref{ass:strong_convexity}, we denote $\ell_\tau$'s potential \change{as} $L(\cdot, \tau)$, such that $\nabla_y L(y, \tau) = \ell_\tau$ and $L$ augments $F$~\eqref{eqn:potential_def} as
\begin{equation}\label{eqn:constrained_lagrangian}
    L(y, \tau) = F(y)  + \tau^\top(Ay - b).
\end{equation}
Given a toll value $\tau$, the toll-augmented game $d(\tau)$ and the tolled MDP Wardrop equilibrium, $y_\tau \in \mc{W}(\ell_\tau)$, are given by 
\begin{equation}\label{eqn:dual_function_and_equilibria}
    d(\tau) = \miny L(y,\tau), \quad y_\tau \in  \argminy L(y, \tau). 
\end{equation}
\ctwo{Under cost vector~\eqref{eqn:augmented_cost}, any feasible affine population constraint will hold for large} values of $\tau$~\cite{us}. We specifically want to compute the minimum toll value to enforce $\mc{C}$~\eqref{eqn:constraint_def} on the MDP Wardrop equilibrium of the toll-augmented game.
\begin{defn}[Minimum toll value]\label{def:minimum_toll_value} \ctwo{Given a constraint set $\mc{C}$~\eqref{eqn:constraint_def}, the minimum toll value $\tau^\star \in \reals_+^{C}$ is the smallest non-negative toll which ensures that the MDP congestion game has constraint-satisfying MDP Wardrop equilibria---i.e., }
\begin{equation}\label{eqn:minimum_toll_def}
   \textstyle \tau^\star  = \min\big\{ \tau \in \reals^C_+ \ | \ \mc{W}(\ell_\tau) \subseteq \mc{C}\big\}.
\end{equation}
\end{defn}
The minimum toll value exists under the following sufficient condition~\cite{us}.
\begin{proposition}~\cite{us}:
When $\mc{C}$ is convex, $\mc{C} \cap \mc{Y}(P, p_0)$ is non-empty, and the cost vector $\ell$ satisfies  Assumption~\ref{ass:strong_convexity}, \change{a unique minimum toll value $\tau^\star$~\eqref{eqn:minimum_toll_def} maximizes $d(\tau)$}.
\begin{equation}\label{eqref:toll_lagrangian_sol}
    \tau^\star = \argmax_{\tau \in \reals_+^C}\Big[\miny L(y,\tau) \Big]= \argmax_{\tau \in \reals_+^C} \ d(\tau).
\end{equation}
\end{proposition}
When $\ell$ is known, \eqref{eqref:toll_lagrangian_sol} directly computes $\tau^\star$.
When $\ell$ is unknown, we cannot explicitly solve for either $\tau^\star$ or $d(\tau)$. 
\begin{problem}\label{prob: 1}
\change{
For MDP congestion games with unknown but strictly increasing congestion costs, find the minimum toll value $\tau^\star$~\eqref{eqn:minimum_toll_def} that ensures the resulting MDP Wardrop equilibrium $y_{\tau^\star}$~\eqref{eqn:dual_function_and_equilibria} satisfies the desired affine constraints $\mc{C}$~\eqref{eqn:constraint_def}.}
\end{problem}

\begin{figure}
    \centering
    \includegraphics[width=\columnwidth]{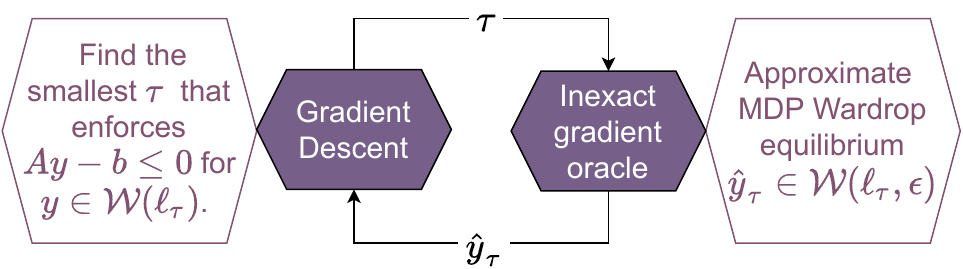}
    \caption{Using approximate MDP Wardrop equilibrium, we perform inexact gradient descent on $\tau$ to find the minimum toll value.}
    \label{fig:planner_dynamics}
\end{figure}
\change{As summarized in Figure~\ref{fig:planner_dynamics}, we compute $\tau^\star$ by querying the $\epsilon$-MDP Wardrop equilibria of $d(\tau)$ and forming an inexact oracle gradient descent.} 
To see how the $\epsilon$-MDP Wardrop equilibrium of a tolled game induces an inexact oracle for $\nabla d(\tau)$, we first derive the analytical expression of $\nabla d(\tau)$. 
\begin{proposition}\label{prop:dual_properties}
If the cost vector $\ell$ satisfies Assumption~\ref{ass:strong_convexity} and $\mc{C}$ satisfies Definition~\ref{def:affineConstraints}, $d$~\eqref{eqn:dual_function_and_equilibria} has the following properties. 
\begin{itemize}
    \item $d$ is concave.
    \item $d$ is $\thickbar{\alpha}$-smooth with $\thickbar{\alpha} = \frac{\norm{A}^2_2}{\alpha}$. I.e., for any $\sigma, \tau \in \reals^{C}$,
    \begin{equation}\label{eqn:d_smoothness}
        \textstyle d(\tau) + \nabla d(\tau)^\top(\sigma - \tau) - \frac{\thickbar{\alpha}}{2}\norm{\sigma - \tau}_2^2 \leq d(\sigma).
    \end{equation}
    \item Let $y_\tau$ be defined as~\eqref{eqn:dual_function_and_equilibria}, then $\nabla d (\tau)$ is given by
    \begin{equation}\label{eqn:dual_gradient}
       \textstyle \nabla d (\tau) = A y_\tau - b.
    \end{equation}
\end{itemize}
\end{proposition}
\ctwo{See \arxiv{App.1.1.1} for proof.} When the costs $\ell$ are unknown, we can compute the $\epsilon$-MDP Wardrop equilibrium via learning algorithms~\cite{zhou2020reinforcement,krichene2015,yu2019primal}.
When applied to tolled games, these $\epsilon$-MDP Wardrop equilibria form $\epsilon$-\emph{inexact oracles} of $d$. 
\begin{defn}[$\epsilon$-inexact oracle]
The $\epsilon$-inexact oracles of $\nabla d(\tau)$ and $d(\tau)$ are given by
\begin{equation}\label{eqn:approximations_to_d}
\thickhat{\nabla}d(\tau) = A \hat{y}_{\tau}(\epsilon) - b, \quad \hat{d}(\tau) = L(\hat{y}_{\tau}(\epsilon), \tau), 
\end{equation}
where $\hat{y}_\tau(\epsilon) \in \mc{W}(\ell_\tau, \epsilon)$~\eqref{eqn:approx_WE} is an $\epsilon$-MDP Wardrop equilibrium satisfying
\begin{equation}\label{eqn:eps_tau_wardrop_eq}
    L(\hat{y}_{\tau}(\epsilon), \tau) \leq L(y_\tau, \tau) + \epsilon.
\end{equation}
\end{defn}
When $\epsilon = 0$, the oracle is exact. When $\epsilon > 0$, the oracle's \ctwo{accuracy directly affects} the concavity and smoothness of $d$ (See~\arxiv{App.1}). 


\section{Minimum toll algorithm}\label{sec:dual_ascent_algorithm}
Convergence of first-order gradient methods relies on the objective's convexity and smoothness. If an inexact gradient preserves concavity and smoothness, \ctwo{its gradient descent will also converge}~\cite{devolder2014first}. In this section, we apply the same concept to tolling in MDP congestion games and analyze on how constraint violation is affected by the $\epsilon$-MDP Wardrop equilibrium. 

\begin{algorithm}[ht!]
\caption{Iterative toll synthesis}
\begin{algorithmic}[1]
\Require \(\ell\), \(P\), \(p_s\), \(\tau_0\).
\Ensure \(\tau^N, y^N\).
\For{\(k=0,1, \ldots\)}
    \State{\( y^{k} \in \mc{W}(\ell + A^\top\tau^k, \epsilon^k)\)}\label{alg:admm_game}
    \State{\(\tau^{k+1} =[ \tau^k + \gamma^k ( A y^{k} - b)]_+\)}\label{algline:dual_ascent_tau_update}
\EndFor 
\end{algorithmic}
\label{alg:dualAscent}
\end{algorithm}

In Algorithm~\ref{alg:dualAscent}, we denote the $k^{th}$ toll charged, the $k^{th}$ $\epsilon$-MDP Wardrop equilibrium, and the $\epsilon$ in the $k^{th}$ $\epsilon$-inexact oracle as $\tau^k$, $y^k$, and $\epsilon^k$, respectively. When $\epsilon^k =0$ $\forall k \in\mathbb{N}$, Algorithm~\ref{alg:dualAscent} is a projected gradient ascent on $d(\tau)$ \ctwo{with sublinear convergence rates~\cite{bubeck2015convex}.}
We analyze Algorithm~\ref{alg:dualAscent}'s convergence when $\epsilon^k > 0$ by through the following quantities,
\begin{equation}\label{eqn:average_quantities}
    \thickbar{\tau}^k = \frac{1}{k}\sum_{s=1}^k \tau^s, \quad \thickbar{y}^k = \frac{1}{k}\sum_{s=0}^{k-1} y^s, \quad E^k = \sum_{s=0}^{k-1} \epsilon^s ,
\end{equation}
where $\thickbar{\tau}^k$/ $\thickbar{y}^k$/$E^k$  is the average toll/average $\epsilon$-MDP Wardrop equilibrium/accumulated $\epsilon$ up to iteration $k$, respectively.
\begin{thm}\label{thm:dualAscent}
If the cost vector $\ell$ satisfies Assumption~\ref{ass:strong_convexity}, and \change{$\gamma \leq \frac{\alpha}{2\norm{A}^2_2}$} for each $k\in\mathbb{N}$, then $\thickbar{\tau}^k$ from~\eqref{eqn:average_quantities} satisfies
\begin{equation}\label{eqn:toll_convergence}
 \textstyle d(\tau^\star) - d(\thickbar{\tau}^k) \leq \frac{1}{k}( \frac{1}{2\gamma}\norm{\tau^0 - \tau^\star}_2^2 + 2E^k), 
\end{equation}
where $\tau^\star$ is the minimum toll value~\eqref{eqn:minimum_toll_def}.
and $E^k$~\eqref{eqn:average_quantities} is the total approximation error.
\end{thm}
See~\arxiv{App.3} for proof. Our proof is inspired by~\cite{devolder2014first,necoara2013rate}. 
\begin{rem}
When $\epsilon^k = \epsilon$ is constant,~\eqref{eqn:toll_convergence} becomes
$\textstyle d(\tau^\star) - d(\thickbar{\tau}^k) \leq \frac{1}{k}( \frac{1}{2\gamma}\norm{\tau^0 - \tau^\star}_2^2) + 2\epsilon.$
\ctwo{Similar to exact gradient descent, $\textstyle\frac{1}{2\gamma}\norm{\tau^0 - \tau^\star}_2^2$ converges sublinearly in $k$. However, the term $2\epsilon > 0$ causes a constant convergence error.}
\end{rem}
Constraint violation of $\thickbar{y}^k$~\eqref{eqn:average_quantities} is similarly bounded. 
\begin{cor}\label{cor:constraint_violation_convergence}
If the cost vector $\ell$ satisfies Assumption~\ref{ass:strong_convexity} and \change{$\gamma \leq \frac{\alpha}{2\norm{A}^2_2}$}, then the constraint violation of the average population distribution $\thickbar{y}^k$ from~\eqref{eqn:average_quantities} satisfies 
\begin{equation}\label{eqn:constraint_violation_convergence}
\textstyle\norm{[A\thickbar{y}^{k} - b]_+}_2 \leq \frac{1}{\gamma k}\Big(\norm{\tau^\star}_2 + \norm{\tau^0 - \tau^\star}_2 + 2\sqrt{\gamma E^k}\Big).    
\end{equation}
\end{cor}
See~\arxiv{App.4} for proof.
\begin{rem}
With a constant error oracle $\epsilon^k = \epsilon$, the average constraint violation will still asymptotically reduce to zero.
\end{rem}
Unlike $\bar{\tau}^k$, $E^k$'s effect on the average constraint violation can be reduced with \emph{larger} step sizes as $2\sqrt{E^k \gamma^{-1}}$. 
\change{We note that Corollary~\ref{cor:constraint_violation_convergence} shows that Algorithm~\ref{alg:dualAscent} is not appropriate for enforcing safety-critical system constraints.}

Algorithm~\ref{alg:dualAscent} also ensures that the average population distribution $\thickbar{y}^k$~\eqref{eqn:average_quantities} converges to the optimal equilibrium for $\tau^\star$.
\begin{thm}\label{thm:primal_convergence}
If the cost vector $\ell$ satisfies Assumption~\ref{ass:strong_convexity} and \change{$\gamma \leq \frac{\alpha}{2\norm{A}^2_2}$}, then the average player population distribution given by  $\thickbar{y}^k$ \eqref{eqn:average_quantities}  satisfies
\begin{equation}\label{eqn:primal_convergence}
    \textstyle\norm{\thickbar{y}^k - y^\star}_2^2 \leq \frac{\alpha}{2\gamma k} D(\tau^0, \tau^\star, E^k),
\end{equation}
where $\tau^\star$ is the minimum toll value, $y^\star$ is the optimal population distribution for $d(\tau^\star)$, and $D(\tau^0, \tau^\star, E^k) $ is given by
\begin{multline}
   \textstyle D(\tau^0, \tau^\star, E^k) = \max\big\{\half\norm{\tau^0}_2^2 + 2E^k, \\
  \textstyle \norm{\tau^\star}_2^2 +\norm{\tau^\star}_2\norm{\tau^0 - \tau^\star}_2 + 2\sqrt{\gamma E^k} \big\}.  
\end{multline}
\end{thm}
See~\arxiv{App.2} for proof. 
\begin{rem}\change{
Similar to $\bar{\tau}^k$, convergence of $\bar{y}^k$ to $y^\star$ is sublinear in $k$ and scales linearly with $E^k$. However, convergence error due to $E^k$ can be minimized by taking \emph{larger} stepsizes.}
\end{rem}
\change{
\noindent\textbf{Fast first-order gradient method.} When $\epsilon^k = \epsilon$ for all $k \in \mathbb{N}$, the fast gradient method~\cite{devolder2014first} augments Algorithm~\ref{alg:dualAscent} with the following update after Step~\ref{algline:dual_ascent_tau_update},
\[\textstyle\tau^{k+1} =  \frac{\norm{A}_2^2}{\alpha (k+3)}\Big[\sum_{i=1}^{k+1} \sqrt{i(i+1)} (A\hat{y}^k - b)\Big]_+ + \frac{k+1}{k+3}\tau^{k+1}. \]
In large networked systems with low-accuracy inexact oracles, the fast gradient method theoretically and empirically diverges from $d(\tau^\star)$~\cite{devolder2014first}. Since its constraint violation results are comparable to Algorithm~\ref{alg:dualAscent}~\cite{necoara2013rate}, we focus on the standard first-order gradient descent instead.
}
\section{Congestion Reduction in Ride-share Networks}\label{sec:ny_rideshare}
\noindent In this section, we model competition among NYC's ride-share drivers as a MDP congestion game and apply Algorithm~\ref{alg:dualAscent} to demonstrate how ride-share companies can implicitly enforce constraints by utilizing tolls.\footnote{Code for the Manhattan ride-share MDP congestion game is available at \href{http://github.com/lisarah/manhattan_MDP_queue_game}{github.com/lisarah/manhattan\_MDP\_queue\_game}.} 
Since origin-destination-specific trip data for ride-share companies are not publicly available, we use the rider demand distribution provided by the NYC TLC as a proxy for Uber's rider demand distribution. In~\cite{taxi_frequency}, the overall rider demand for TLC is estimated to be about $40\%$ of the rider demand for Uber. 
\subsection{Ride-sharing MDP Game \ctwo{With Queues}} \label{sec:ride_share_game}
We consider a cohort of competitive ride-share drivers in Manhattan, NYC repeatedly operating between $9$ am and noon. Using six hundred thousand trip data from the yellow taxi data during January, 2019~\cite{new_york_taxi}, we model individual driver as a finite time horizon MDP \ctwo{in a queuing network}. 

\change{
\noindent\textbf{Modeling assumptions}. 1) \ctwo{All trips take discretized times of $\{15, 20, 30,\ldots\}$ based on the trip distance. 2) The initial driver distribution is uniform across all Manhattan zones.} We find that varying the initial distribution does not significantly impact the time-averaged MDP Wardrop equilibrium or the toll norm, as long the constraints are satisfied.  3) From~\cite{uber_driver_density}, the Uber driver population in NYC is approximately $50 000$. We assume that $20\%$ of the total population works in Manhattan between $9$ am and noon. 
}

\noindent\textbf{States}. \ctwo{Each state is given by $s = (z, q)$, where $z \in [63]$ is one of $63$ Manhattan zones, visualized the right plot in Figure~\ref{fig:manhattan_taxi_zones} (islands excluded), and $q \in [Q]$, where $Q = 7$ is the queue level. If $q=0$, the driver is in zone $z$ without a rider. If $q > 0$, the driver is $q$ time steps away from completing a ride to zone $z$.}  \ctwo{The set of geographically adjacent(sharing one or more edges) zones of $z$ is given by $\mc{N}(z)$}.

\noindent\textbf{Actions}. 
\ctwo{The action set of state $(z, q)$ is $q$-dependent and is given by $\mc{A}(z, q)$.} 
\ctwo{When $q > 0$, the driver is completing a ride. Therefore, the only action, $a_z$, is to finish the ride and $\mc{A}(z,q)$ is a singleton set $\mc{A}(z, q) = \{a_{z}\}, \forall z \in [63], \ q \geq 1.$
When $q = 0$, the driver can either go to a neighboring zone ($a_{z'}$) or pick up a rider in the current zone ($a_z$). The action set of $(z,0)$ is $\mc{A}(z, 0) = \big\{a_{z'} | \ z' \in \mc{N}(z) \cup \{z\}\big\}, \ \forall z \in [63]. $
We can extend the $q$-dependent action model to conform to the MDP model in Section~\ref{sec:prelim} by defining $A =   \max_{z, q}|\mc{A}(z, q)|$, and for all $(z,q)$ where $|\mc{A}(z,q)| < A$, insert $A - |\mc{A}(z,q)|$  actions with infinite costs.}

\noindent\textbf{Time.} The average trip time from the TLC data is $12.02$ minutes. \ctwo{We add buffer time for drivers to locate and drop off riders, such that the MDP time interval is} $15$ minutes between $9$ am and noon for a total of \ctwo{$T = 12$} time steps. 

\noindent\textbf{Transition Dynamics.}
\ctwo{Transition dynamics are $q$-dependent. When $q > 0$, the driver is completing a ride $(a_z)$. Then, for all $z \in [63]$ and $t \in [T]$, the transition dynamics of $(z,q, a_z)$ is given by
\begin{equation}
    P(t, s', a_{z}, z, q)  = \begin{cases}
    1 & s' = (z, q-1) \\
    0 & \text{otherwise}
    \end{cases}, \forall \  q \geq 1.
\end{equation}
When $q = 0$, the driver may go to an adjacent zone $(\{a_{z'} |z' \in \mc{N}(z)\} )$ or pick up a rider $(a_z)$. For $a_{z'}$, the transition dynamics is given by
\begin{equation}
    P(t, s', a_{z'}, z, q) =
    \left\{\begin{array}{ll}
  1 - \delta, & \text{if}\ s = (z', 0),\\
  \frac{\delta}{| \mc{N}(z) | - 1}, & \text{if}\ s = (\bar{z}, 0), \bar{z} \in \mc{N}(z)/\{{z}\},\\
  0,& \text{otherwise},
  \end{array}\right. 
\end{equation}}
where $\delta \in [0,1)$ models the driver's probability of real-time deviation from a chosen strategy. We set $\delta = 0.01$.

\ctwo{For action $a_z$ from state $(z, 0)$, drivers will find a ride and transition to the appropriate queue in the destination zone.} 
The transition dynamics for $(z,q, a_z, t)$ is derived using the TLC ride demand distribution at \ctwo{$(z,0)$~\cite{new_york_taxi}. Let $N(z, z', q, t)$ be the number of trips with origin-destination $(z, z')$ at time step $t$ that took between $15q$ and $15(q+1)$ minutes. For all $z' \in [63]$ and $q \in [Q]$ at time $t\in [T]$, the transition probability to state $(z', q)$ is given by
\[ P((z', q), (z, 0), a_{z}, t) = \frac{N(z, z', q, t)}{\sum_{\bar{q} \in [Q]}\sum_{\bar{z} \in [63]} N(z, \bar{z}, \bar{q}, t)}.\]}
\ctwo{
Note that $z'$ need not be an adjacent zone to $z$.}

\noindent\textbf{Driver costs}. \ctwo{Driver costs are $q$-dependent at each state $s = (z, q)$. When $q > 0$, the driver cost is given by
\[\ell_{tsa}(y_{tsa}) = \beta y_{tsa}, \ \forall \ t \in [T], \ a \in \mc{A}(z,q),\]
where $\beta$ models the minor congestion effect of drivers entering zone $z$ at queue level $q$. We set  $\beta = 0.001$.} 

\ctwo{When $q = 0$, we follow the cost model in~\cite{us}, given by}
\begin{align}\label{eqn:0_queue_cost}
\ell_{tsa}(y_{tsa}) & =\mathbb{E}_{s'} \left[c_{ts's}^\text{trav} - m_{ts's} \right] + c_{t}^\text{wait}\cdot y_{tsa} \\
& = \textstyle\sum_{s'}P_{ts'sa}\left[ c_{ts's}^\text{trav} -m_{ts's}\right] + c_{t}^\text{wait}\cdot y_{tsa} \notag
\end{align}
\ctwo{The parameters in~\eqref{eqn:0_queue_cost} are action-dependent: $m_{ts's}$ is the monetary reward, $c^{wait}_t$ is the congestion scaling coefficient, and $c_{ts's}^\text{trav}$ is the fuel cost.
\begin{enumerate}
    \item For $a_{z'}$ and $s' = (z', 0)$,  $m_{ts's} = 0$ and $c^{wait}_t = 0.01$.  The term $c_{ts's}^\text{trav}$ is the fuel cost of reaching $z'$, given by
    \begin{equation}\label{eqn:travel_cost}
        c_{ts's}^\text{trav} = \mu
        \underbrace{d_{zz'}}_{\text{mi}}
        \underbrace{\big(\text{Vel}\big)^{-1}}_{\text{hr}/\text{mi}} + 
        \underbrace{\big(\substack{\text{Fuel Price} } \big)}_{\$/\text{gal}}
        \underbrace{\big( \substack{\text{Fuel Eff}} \big)^{-1}}_{\text{gal}/\text{mi}} 
        \underbrace{d_{zz'}}_{\text{mi}}. 
    \end{equation}
    The parameter $d_{zz'}$ is the estimated trip distance between $z$ and $z'$. When $z' = z$, $d_{zz}$ is the average distance (mi) for all  $(z, z)$ trips from the TLC data. When $z' \neq z$, $d_{zz'}$ is the Haversine distance (mi) between $z$ and $z'$. The parameter $\mu$ is a time-money tradeoff parameter, given in Table~\ref{tab:params} along with other parameters. 
    \begin{table}[H]
        \begin{center}
        \centering
        \begin{tabular}{|cccc|}
        \hline
        $\mu$  & Velocity & Fuel Price & Fuel Eff   \\
        \hline
        \$15 /mi & 8 mph & \$2.5/gal & 20 mi/gal \\ \hline
        \end{tabular}
        \centering
        \caption{Parameters for the driver cost function.}
        \label{tab:params}
        \end{center}
    \end{table}
    \item For $a = a_z$ and $s' = (z', 0)$, $m_{ts's}$ is the monetary reward, defined using Uber's NYC pay rate~\cite{uber_pricing}  as 
    \begin{equation}\label{eqn:trip_costs}
     \textstyle m_{ts's}  = \max\Big( \$ 7, \$ 2.55  + \$ 0.35 \cdot \Delta t+ \$ 1.75 \cdot \Delta d\Big),
    \end{equation}
    where $\Delta t$ is the trip time (min) and $\Delta d$ is the trip distance (mi). We set $\Delta t = 12$ as the average trip time from the TLC data and $\Delta d$ as the estimated Haversine distance between $(z, z')$. The parameter $c_{t}^\text{wait}$ is the coefficient of congestion, scaled linearly by the portion of drivers who are waiting for a rider, and is given by
    \begin{align}
       \textstyle c^\text{wait}_{tsa}  = \underset{s'}{\mathbb{E}}\left[m_{ts's}\right] \cdot 
        \Big(
        \underbrace{
        \substack{\text{Customer Demand Rate} }
        }_{\text{rides}/\Delta t}
        \Big)^{-1},
        \label{eq:travel}
    \end{align}
    where $m_{ts's}$ is given by~\eqref{eqn:trip_costs}  and the customer demand rate is derived from TLC data per time interval per day. We estimate the Uber ride demand to be 2.5 times more than the Yellow Taxi's ride demand in January 2019~\cite{uber_vs_taxi_demand} and scale the TLC data accordingly. 
\end{enumerate}}
\subsection{Online learning via conditional gradient descent}
\ctwo{When drivers optimize their strategies for the ride-share game in Section~\ref{sec:ride_share_game}, we assume that they cannot directly access the congestion costs model and transition dynamics.} Instead, they collectively receive costs for a chosen joint policy, and iterate to find the equilibrium policy. 

We implement the learning method from~\cite[Alg.3]{us}. Inspired by conditional gradient descent (Frank-Wolfe), \cite[Alg.3]{us} implicitly enforces $y \in \mc{Y}(P, p_0)$ by solving linearized game potentials~\eqref{eqn:integral_potential} via dynamic programming. Frank-Wolfe converges rapidly to low-accuracy solutions. Based on Frank-Wolfe's stopping criterion, the $\epsilon$ in the $\epsilon$-MDP Wardrop equilibrium is given by 
\begin{equation}\label{eqn:manhattan_stopping_criterion}
  \textstyle\epsilon^k = \big(\ell(y^k) + A^\top\tau^k\big)^\top(y^k - y^{k+1}).
\end{equation} 
\begin{figure}
    \centering
    \includegraphics[width=1.0\columnwidth]{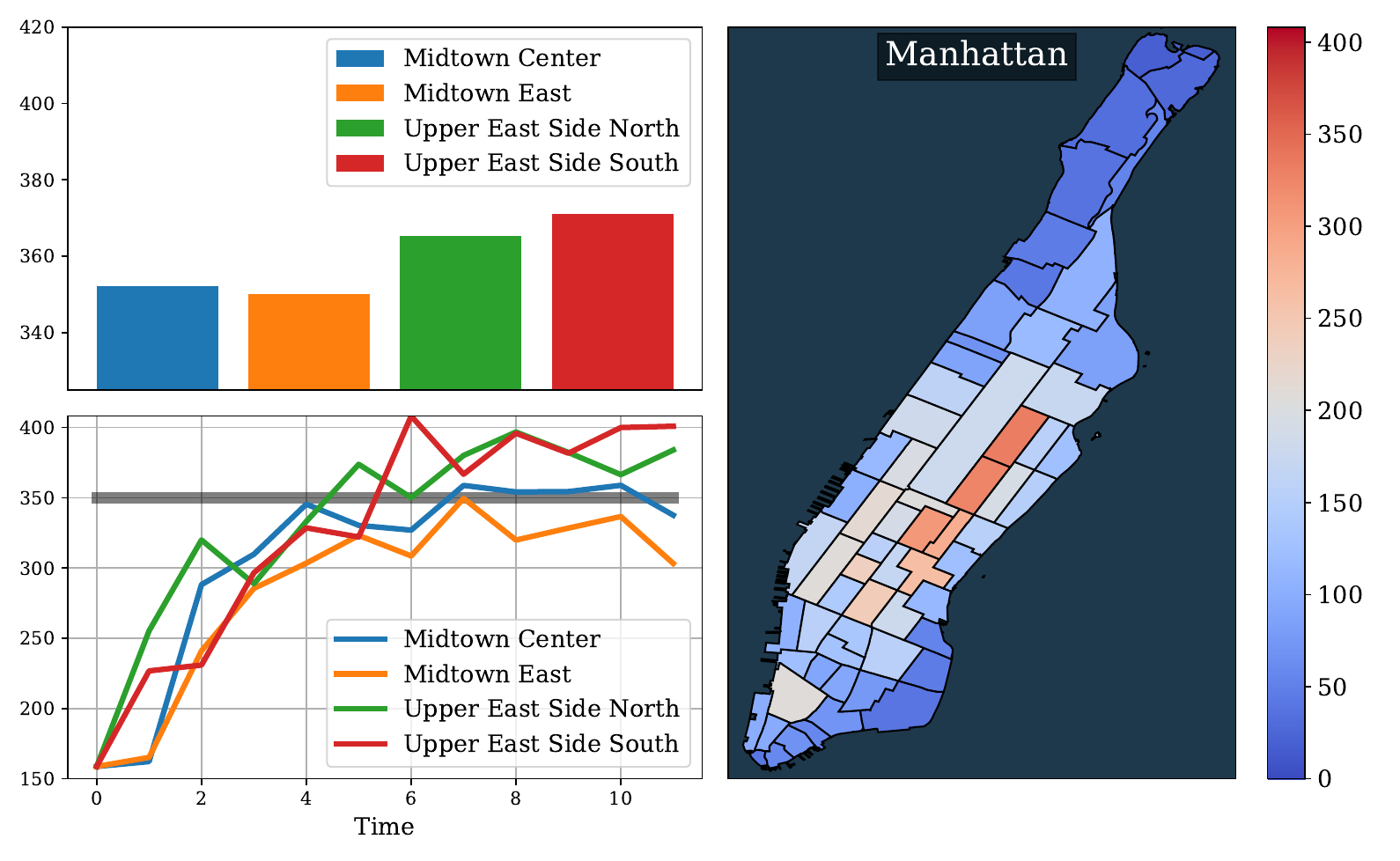}
    \caption{Predicted ride-share traffic in Manhattan.}
    \label{fig:manhattan_taxi_zones}
\end{figure}
We set $\epsilon^k = 1e^3$, which is approximately equal to a normalized error of $0.5\%$ for the unconstrained game potential. The corresponding $\epsilon$-MDP Wardrop equilibrium and the driver densities of the most congested zones are shown in Figure~\ref{fig:manhattan_taxi_zones}. 
\subsection{Reducing driver presence in congested taxi zones}
Suppose the ride-share company wishes to reduce \ctwo{the driver density in congested zones to below \ctwo{$350$} per zone per time step via Algorithm~\ref{alg:dualAscent}.} This constraint can be formulated as
\begin{figure}
    \centering
    \includegraphics[width=\columnwidth]{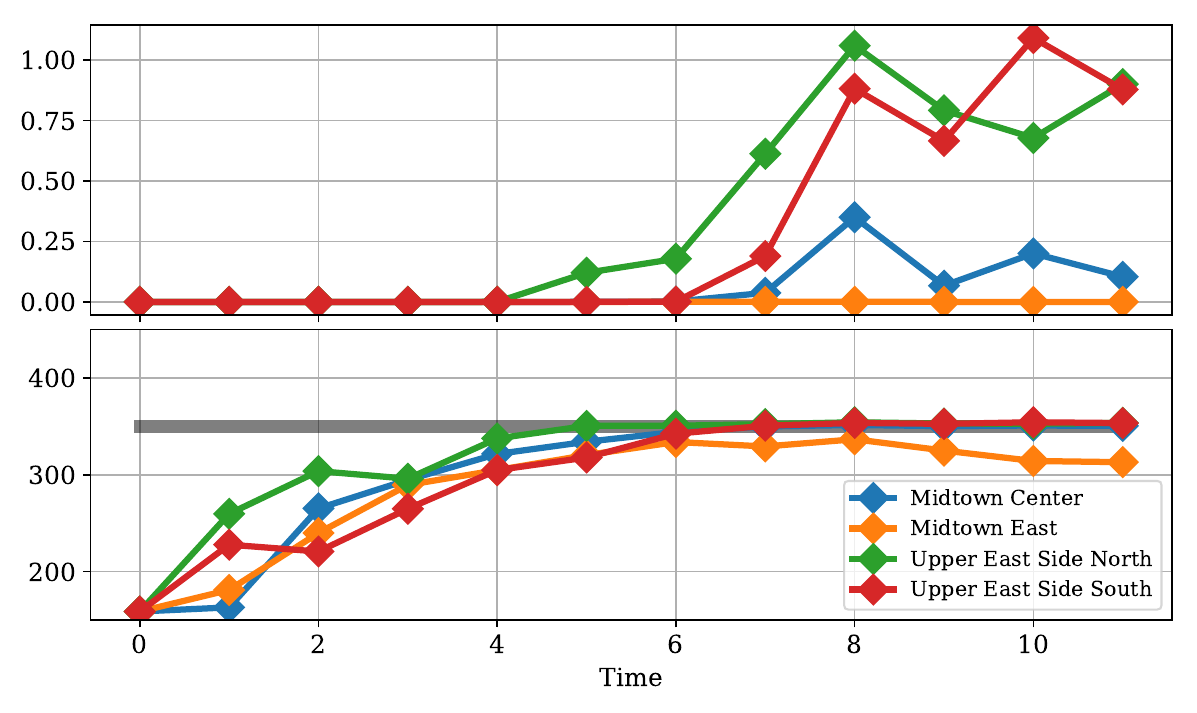}
    \caption{Manhattan game under constraint~\eqref{eqn:manhattan_constraints}. The top line plot shows congestion tolls (\$/$15$min). The bottom line plot shows the congested driver distributions (drivers/$15$min). }
    \label{fig:manhattan_tolled}
\end{figure}
\begin{equation}\label{eqn:manhattan_constraints}
    \textstyle \sum_{a} y_{tsa} \leq 350, \ s = (z, 0), \ \forall \ (t, z) \in [T+1]\times[63].
\end{equation}
Each $(t,z)\in [T+1]\times[63]$ corresponds to a constraint $A_i \in \reals^{(T+1)SA}$, where for all $a \in \mc{N}(z, 0)$, the $(t,(z,0),a)^{th}$ entry is $1$ and all other entries are $0$. Thus, we enforce a total of $63 \times 12 = 1008$ constraints of form~\eqref{eqn:manhattan_constraints}. The constraint matrix is $\textstyle A = [A_1, \ldots, A_{(T+1)S}]^\top \in \reals^{(T+1)S\times (T+1)SA}$. 
\subsection{Discussion}
We run Algorithm~\ref{alg:dualAscent} for $2000$ iterations at $\epsilon^k = 0.5\%$ of the unconstrained potential value. The results are shown in Figure~\ref{fig:manhattan_tolled}. The resulting constraint violation has $2$ norm $10.24$ for the whole time horizon. The tolls $2$ norm is $2.94$.

In Figure~\ref{fig:manhattan_tolled}(top), we see that for tolls around $\$1$ per time step per state, we can decrease the average constraint violation from over $200$ drivers to less than $10$ drivers for whole time horizon. This is comparable to the proposed toll value for lower Manhattan ($\$2.25$ per entry)~\cite{nyc_congestion_toll}. \ctwo{In Figure~\ref{fig:manhattan_system_analysis} left, we evaluate Algorithm~\ref{alg:dualAscent} by its toll value convergence and the constraint violation during tolling. Note that the average constraint violation $\norm{[A\thickbar{y}^k - b]_+}_2$ differs from the last-iterate constraint violation $\norm{[A{y}^k - b]_+}_2$, and the last iterate constraint violation does not converge in part due to drivers' imprecision in finding equilibrium strategy. . }  

\ctwo{\textbf{Social cost}.} A major concern is the effects of tolling on the average drivers' earnings, measured by the \emph{social cost}~\cite{us}. When the social cost increases significantly, drivers may quit,  thus reducing the ride-share workforce.
We show empirically in Figure~\ref{fig:manhattan_system_analysis} (right) that tolling does not significantly impact driver earnings: the average driver earnings during the tolling process are normalized against the untolled average earnings. \ctwo{During tolling,  the social cost in fact decreased, implying that the average driver earnings increased. Therefore, congestion-based tolling is unlikely to cause quitting among the driver population.} 
\begin{figure}
\center
\includegraphics[width=\columnwidth]{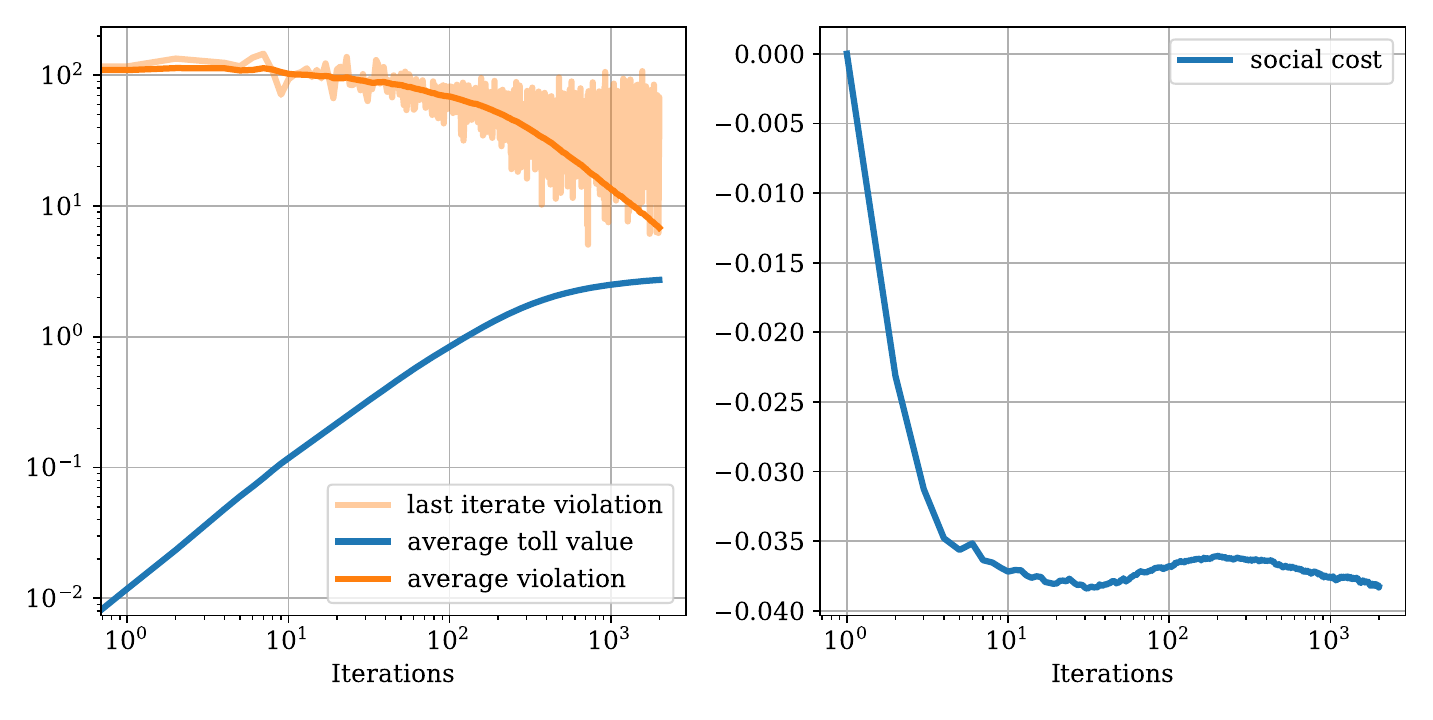}
\caption{Left: average toll value and constraint violation during toll synthesis. Right: average driver earnings during toll synthesis.}
    \label{fig:manhattan_system_analysis}
\end{figure}

\ctwo{\textbf{Equilibrium accuracy}.} The effect of $\epsilon^k$ on the toll value $\norm{\thickbar{\tau}}_2$ and the average constraint violation $\norm{[A\thickbar{y}^k - b]_+}_2$ is shown in  Figure~\ref{fig:epsilon_analysis} after $1000$ iterations for \ctwo{$\epsilon = [100, 1000, 5000, 10000, 50000]$.} The increased accuracy in $\epsilon$ decreases both the toll value and the constraint violation during the tolling process, thus providing more incentive to accurately compute the minimum toll value. 
\begin{figure}
\center
\includegraphics[width=\columnwidth]{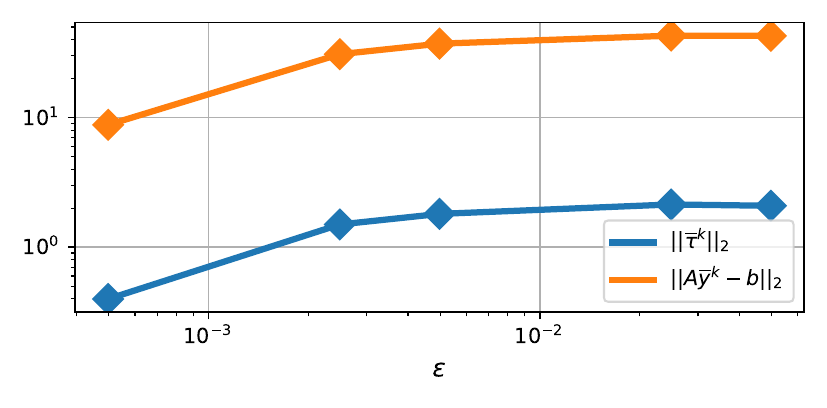}
\caption{$\epsilon$ vs average toll and constraint violation at $k=1000$.}
\label{fig:epsilon_analysis}
\end{figure}
\section{Conclusion}
We presented an iterative tolling method that allows system-level operators to enforce constraints on a MDP congestion game with unknown congestion costs. We showed that an $\epsilon$-MDP Wardrop equilibrium corresponds to an inexact gradient oracle of the tolled game, and derived conditions for convergence of the inexact gradient descent problem. We applied our results the ride-share system in Manhattan, NYC. Future extensions to this work include extending the toll synthesis method to work for general convex constraints.

\bibliographystyle{plain}        
\bibliography{reference}           

\begin{thebibliography}{10}

\bibitem{uber_pricing}
{Alvia}.
\newblock Uber new york.
\newblock {http://www.alvia.com/uber-city/uber-new-york/}, 2021.
\newblock Accessed: 2021-02-14.

\bibitem{beckmann1952continuous}
Martin Beckmann.
\newblock A continuous model of transportation.
\newblock {\em Econometrica}, pages 643--660, 1952.

\bibitem{bertsekas1999nonlinear}
Dimitri~P Bertsekas.
\newblock {\em Nonlinear Programming}.
\newblock Athena Scientific Belmont, 1999.

\bibitem{bubeck2015convex}
S{\'e}bastien Bubeck et~al.
\newblock Convex optimization: Algorithms and complexity.
\newblock {\em Found. Trends Mach. Learn.}, 8(3-4):231--357, 2015.

\bibitem{uber_vs_taxi_demand}
Nicu Calcea.
\newblock Nycdot’s experience with big data and use in transportation
  projects.
\newblock
  {https://citymonitor.ai/transport/uber-lyft-rides-during-coronavirus-pandemic-taxi-data-5232},
  2017.
\newblock Accessed: 2021-02-14.

\bibitem{calderone2017markov}
Dan Calderone and S~Shankar Sastry.
\newblock Markov decision process routing games.
\newblock In {\em Proc. Int. Conf. Cyber-Physical Syst.}, pages 273--279. ACM,
  2017.

\bibitem{demir2015decentralized}
Nazl{\i} Demir, Utku Eren, and Beh{\c{c}}et A{\c{c}}{\i}kme{\c{s}}e.
\newblock Decentralized probabilistic density control of autonomous swarms with
  safety constraints.
\newblock {\em Auton. Robots}, 39(4):537--554, 2015.

\bibitem{devolder2014first}
Olivier Devolder, Fran{\c{c}}ois Glineur, and Yurii Nesterov.
\newblock First-order methods of smooth convex optimization with inexact
  oracle.
\newblock {\em Math. Program.}, 146(1):37--75, 2014.

\bibitem{nyc_congestion_toll}
Lauren~Aratani Erin~Durkin.
\newblock New york becomes first city in us to approve congestion pricing.
\newblock
  {https://www.theguardian.com/us-news/2019/apr/01/new-york-congestion-pricing-manhattan}.
\newblock Accessed: 2021-02-14.

\bibitem{etesami2020smart}
S~Rasoul Etesami, Walid Saad, Narayan~B Mandayam, and H~Vincent Poor.
\newblock Smart routing of electric vehicles for load balancing in smart grids.
\newblock {\em Automatica}, 120:109148, 2020.

\bibitem{fazlyab2018distributed}
Mahyar Fazlyab, Santiago Paternain, Alejandro Ribeiro, and Victor~M Preciado.
\newblock Distributed smooth and strongly convex optimization with inexact dual
  methods.
\newblock In {\em Amer. Control Conf.}, pages 3768--3773. IEEE, 2018.

\bibitem{krichene2015}
W~Krichene, B~Drigh{\`{e}}s, and A~Bayen.
\newblock {Online Learning of Nash Equilibria in Congestion Games}.
\newblock {\em SIAM J. Control Optim.}, 53(2):1056--1081, 2015.

\bibitem{lasry2007mean}
Jean-Michel Lasry and Pierre-Louis Lions.
\newblock Mean field games.
\newblock {\em Jpn. J. Math.}, 2(1):229--260, 2007.

\bibitem{us}
Sarah~HQ Li, Yue Yu, Daniel Calderone, Lillian Ratliff, and Beh{\c{c}}et
  A{\c{c}}{\i}kme{\c{s}}e.
\newblock Tolling for constraint satisfaction in markov decision process
  congestion games.
\newblock In {\em 2019 American Control Conference (ACC)}, pages 1238--1243.
  IEEE, 2019.

\bibitem{lin2018environmental}
X~Lin.
\newblock Environmental constraints in urban traffic management: Traffic
  impacts and an optimal control framework.
\newblock 2018.

\bibitem{uber_driver_density}
Aarian Marshall.
\newblock New york city flexes again, extending cap on uber and lyft.
\newblock
  {https://www.wired.com/story/new-york-city-flexes-extending-cap-uber-lyft/}.
\newblock Accessed: 2021-02-14.

\bibitem{miguelez2004practical}
Enrique~Lobato Migu{\'e}lez, Luis~Rouco Rodr{\'\i}guez, TGS Roman,
  FM~Echavarren Cerezo, Ma~Isabel~Navarrete Fern{\'a}ndez, Rosa~Casanova
  Lafarga, and Gerardo~L{\'o}pez Camino.
\newblock A practical approach to solve power system constraints with
  application to the spanish electricity market.
\newblock {\em IEEE Trans. Power Syst.}, 19(4):2029--2037, 2004.

\bibitem{monderer1996potential}
Dov Monderer and Lloyd~S Shapley.
\newblock Potential games.
\newblock {\em Games Economic Behavior}, 14(1):124--143, 1996.

\bibitem{necoara2013rate}
Ion Necoara and Valentin Nedelcu.
\newblock Rate analysis of inexact dual first-order methods application to dual
  decomposition.
\newblock {\em IEEE Transactions on Automatic Control}, 59(5):1232--1243, 2013.

\bibitem{nesterov2005smooth}
Yu~Nesterov.
\newblock Smooth minimization of non-smooth functions.
\newblock {\em Mathematical programming}, 103(1):127--152, 2005.

\bibitem{new_york_taxi}
City of~New~York.
\newblock Tlc trip record data.
\newblock {https://www1.nyc.gov/site/tlc/about/tlc-trip-record-data.page}.
\newblock Accessed: 2021-02-14.

\bibitem{ostrovskii2021efficient}
Dmitrii~M Ostrovskii, Andrew Lowy, and Meisam Razaviyayn.
\newblock Efficient search of first-order nash equilibria in nonconvex-concave
  smooth min-max problems.
\newblock {\em SIAM J. Optim.}, 31(4):2508--2538, 2021.

\bibitem{patriksson2015traffic}
Michael Patriksson.
\newblock {\em The Traffic Assignment Problem: Models and Methods}.
\newblock Courier Dover Publications, 2015.

\bibitem{puterman2014markov}
Martin~L Puterman.
\newblock {\em Markov Decision Processes: Discrete Stochastic Dynamic
  Programming}.
\newblock John Wiley \& Sons, 2014.

\bibitem{ratliff2018adaptive}
Lillian~J Ratliff and Tanner Fiez.
\newblock Adaptive incentive design.
\newblock {\em arXiv preprint arXiv:1806.05749 [cs.GT]}, 2018.

\bibitem{rosenthal1973class}
Robert~W Rosenthal.
\newblock A class of games possessing pure-strategy nash equilibria.
\newblock {\em International Journal of Game Theory}, 2(1):65--67, 1973.

\bibitem{roth2016watch}
Aaron Roth, Jonathan Ullman, and Zhiwei~Steven Wu.
\newblock Watch and learn: Optimizing from revealed preferences feedback.
\newblock In {\em Proc. Ann. ACM Symp. Theory Comput.}, pages 949--962. ACM,
  2016.

\bibitem{taxi_frequency}
Todd Schneider.
\newblock Taxi and ridehailing usage in new york city.
\newblock
  {https://toddwschneider.com/dashboards/nyc-taxi-ridehailing-uber-lyft-data/},
  2021.
\newblock Accessed: 2021-02-14.

\bibitem{shapley1953stochastic}
Lloyd~S Shapley.
\newblock Stochastic games.
\newblock {\em Proc. Nat. Acad. Sci.}, 39(10):1095--1100, 1953.

\bibitem{swamy2007effectiveness}
Chaitanya Swamy.
\newblock The effectiveness of stackelberg strategies and tolls for network
  congestion games.
\newblock In {\em Proc. ACM-SIAM Symp. Discrete Algorithms}, pages 1133--1142,
  2007.

\bibitem{surge_pricing}
Uber.
\newblock How surge pricing works.
\newblock {https://www.uber.com/us/en/drive/driver-app/how-surge-works/}, 2021.
\newblock Accessed: 2021-04-27.

\bibitem{yu2019primal}
Yue Yu, Dan Calderone, Sarah~HQ Li, Lillian~J Ratliff, and Beh{\c{c}}et
  A{\c{c}}{\i}kme{\c{s}}e.
\newblock Variable demand and multi-commodity flow in markovian network
  equilibrium.
\newblock {\em Automatica}, 140:110224, 2022.

\bibitem{zhou2020reinforcement}
Bo~Zhou, Qiankun Song, Zhenjiang Zhao, and Tangzhi Liu.
\newblock A reinforcement learning scheme for the equilibrium of the in-vehicle
  route choice problem based on congestion game.
\newblock {\em Applied Mathematics and Computation}, 371:124895, 2020.

\end{thebibliography}
\appendix
\setcounter{equation}{0}
\renewcommand\theequation{A.\arabic{equation}}
\section{Appendix}
\subsection{Properties of the tolled game}
To supplement our main results, we provide concavity and smoothness properties of the tolled game $\hat{d}(\tau)$ with respect to $\epsilon$-MDP Wardrop equilibria $\hat{y}_\tau(\epsilon)$.
\begin{lem}[Concavity]\label{lem:approx_concavity}
Under Assumption~\ref{ass:strong_convexity}, all $\epsilon$-MDP Wardrop equilibria $\hat{y}_{\tau}(\epsilon)$ given by~\eqref{eqn:eps_tau_wardrop_eq} will generate $\epsilon$-inexact oracles~\eqref{eqn:approximations_to_d} that satisfy
\[d(\sigma) \leq \hat{d}(\tau) + \thickhat{\nabla} d(\tau)^\top(\sigma - \tau), \quad \forall \ \sigma, \ \tau \in \reals^C_+. \]
\end{lem}
\begin{pf}
We denote $\hat{y}_{\tau}(\epsilon)$ by $\hat{y}_\tau$ for simplicity.
Since \(\hat{y}_{\tau}\in\mc{W}(\ell_\tau, \epsilon)\subset\mc{Y}(P, p_0)\), using \eqref{eqn:dual_function_and_equilibria} we can show that
\begin{equation}\label{eqn:lem1 eq 1}
    d(\sigma)\leq L(\hat{y}_{\tau}, \sigma).
\end{equation}
Combining~\eqref{eqn:lem1 eq 1} with the fact that \(L(\hat{y}_{\tau}, \sigma)=L(\hat{y}_{\tau}, \tau)+\thickhat{\nabla}d(\tau)^\top (\sigma-\tau)\), we obtain Lemma~\ref{lem:approx_concavity}.
\end{pf}

\begin{lem}[$\epsilon$-approximate smoothness]\label{lem:approx_smoothness} Under Assumption~\ref{ass:strong_convexity}, all $\epsilon$-MDP Wardrop equilibria $\hat{y}_{\tau}(\epsilon)$~\eqref{eqn:eps_tau_wardrop_eq} will generate inexact oracles~\eqref{eqn:approximations_to_d} that satisfy
\begin{multline}\label{eqn:pseudo_smoothness}
\hat{d}(\tau) + \thickhat{\nabla} d(\tau)^\top(\sigma - \tau) -   \textstyle \frac{\norm{A}_2^2}{\alpha} \norm{\sigma - \tau}_2^2\leq d(\sigma) + 2\epsilon,\\
\change{\forall \ \tau, \sigma \in \reals^{C}}.
\end{multline}
\end{lem}
\begin{pf}
We denote $\hat{y}_{\tau}(\epsilon)$ by $\hat{y}_\tau$ for simplicity and recall $y_\tau$ from~\eqref{eqn:dual_function_and_equilibria}. From Proposition~\ref{prop:dual_properties}, we know that
\begin{equation}\label{eqn: lem1 eqn1}
    \nabla d(\tau)=\thickhat{\nabla}d(\tau)+A(y_\tau-\hat{y}_\tau).
\end{equation}
Substituting \eqref{eqn: lem1 eqn1} into \eqref{eqn:d_smoothness}, we obtain the following
\begin{equation}\label{eqn: lem1 eqn2}
   \begin{aligned}
   0 \leq &  d(\sigma) - d(\tau) - \textstyle\thickhat{\nabla}d(\tau)^\top(\sigma - \tau) +  \frac{\norm{A}_2^2}{2\alpha}\norm{\sigma - \tau}_2^2\\
    &- \big(A(y_\tau - \hat{y}_\tau)\big)^\top(\sigma - \tau). 
\end{aligned} 
\end{equation}
Furthermore, we can show
\begin{equation}\label{eqn: lem1 eqn3}
   \begin{aligned}
    \Big|\big(A(y_\tau - \hat{y}_\tau)\big)^\top(\sigma - \tau)\Big| \leq &\norm{\hat{y}_\tau - y_\tau}_2\cdot\norm{A}_2 \cdot\norm{\sigma - \tau}_2 \\
    \leq & \textstyle\frac{\alpha}{2}\norm{\hat{y}_\tau - y_\tau}_2^2 +  \frac{\norm{A}_2^2}{2{\alpha}}\norm{\sigma - \tau}_2^2,
\end{aligned} 
\end{equation}
where the first step is due to the Cauchy–Schwarz inequality, and the second step is due to the inequality of arithmetic and geometric inequalities.

Next, we note that $F$~\eqref{eqn:potential_def} and subsequently $L(y, \tau)$~\eqref{eqn:constrained_lagrangian} is strongly convex under Assumption~\ref{ass:strong_convexity}. We combine this with the fact that $L(y_\tau, \tau)=d(\tau)$ from~\eqref{eqn:dual_function_and_equilibria} to obtain 
\begin{equation}\label{eqn: lem1 eqn4}
    \textstyle\frac{\alpha}{2}\norm{\hat{y}_\tau - y_\tau}_2^2\leq L(\hat{y}_\tau, \tau)-d(\tau).
\end{equation}
From~\eqref{eqn:eps_tau_wardrop_eq}, $\hat{y}_\tau$ satisfies
\begin{equation}\label{eqn: lem1 eqn5}
    L(\hat{y}_\tau, \tau) - d(\tau) \leq \epsilon.
\end{equation}
Summing up~\eqref{eqn:eps_tau_wardrop_eq},~\eqref{eqn: lem1 eqn2}, \eqref{eqn: lem1 eqn3}, \eqref{eqn: lem1 eqn4}, and $2\times$\eqref{eqn: lem1 eqn5}, we obtain~\eqref{eqn:pseudo_smoothness}, which completes the proof. 
\end{pf}
\begin{lem}\label{lem:residue_inequality}
Under Assumption~\ref{ass:strong_convexity}, if \change{$\gamma\leq \textstyle\frac{\alpha}{2\norm{A}_2^2}$}, $\tau^s$ from Algorithm~\ref{alg:dualAscent} satisfies
\begin{equation}\label{eqn:residue_inequality}
\begin{aligned}
    \norm{\tau^{s+1} - \tau}_2^2 \leq &\norm{\tau^s - \tau}_2^2 + 2\gamma\big(d(\tau^{s+1}) - L({y}^s, \tau^s) + 2\epsilon^s \\
   & + \thickhat{\nabla}d(\tau^s)^\top(\tau^s - \tau)\big), \quad \forall \ \tau \in \reals^C_+, \ k \geq 0.
\end{aligned}
\end{equation}
\end{lem}
\begin{pf}
Given $\tau \in \reals_+^C$, let $r^s = \norm{\tau^s - \tau}_2^2$. We compute $r^{s+1} - r^s$ using the law of cosine as  
\begin{equation}\label{eqn: lem3 eqn1}
\begin{aligned}
 r^{s+1}- r^s = &  2(\tau^{s+1} - \tau^s)^\top(\tau^{s+1} - \tau) - \norm{\tau^{s+1} - \tau^s}_2^2.
\end{aligned}
\end{equation}
From line~\ref{algline:dual_ascent_tau_update} of Algorithm~\ref{alg:dualAscent},   \(\tau^{s+1} = [\tau^s+\gamma(Ay^s-b)]_+\). Using~\cite[Lem 3.1]{bubeck2015convex}, the projection  onto \(\mathbb{R}_+^C\) implies that
\begin{equation}\label{eqn: lem3 eqn2}
    0\leq (\tau^s+\gamma\thickhat{\nabla }d(\tau^s)-\tau^{s+1})^\top(\tau^{s+1}-\tau)
\end{equation}
From~\eqref{eqn: lem3 eqn2}, we can upper bound $(\tau^{s+1} - \tau^s)^\top(\tau^{s+1} - \tau) $ and combine with~\eqref{eqn: lem3 eqn1} to obtain
\begin{equation}\label{eqn: lem 3 eqn 3}
\begin{aligned}
& \textstyle r^{s+1} - r^s  \leq \textstyle  2\gamma\thickhat{\nabla} d(\tau^s)^\top(\tau^{s+1} - \tau) - \norm{\tau^{s+1} - \tau^s}_2^2
\end{aligned}
\end{equation}
From Lemma~\ref{lem:approx_smoothness}, we recall
\begin{equation}\label{eqn: lem 3 eqn 4}
\begin{aligned}
  & L(y^s, \tau^s) - d(\tau^{s+1})  - 2\epsilon^s \\
  \leq  &d(\tau^s)^\top(\tau^{s}-\tau^{s+1})  +\textstyle\frac{\norm{A}_2^2}{\alpha}\norm{\tau^{s+1} - \tau^s}_2^2 .
\end{aligned}
\end{equation}
We can then combine~\eqref{eqn: lem 3 eqn 3} and $2\gamma\times$\eqref{eqn: lem 3 eqn 4} to derive 
\begin{equation}\label{eqn: lem 3 eqn 5}
\begin{aligned}
   &r^{s+1} - r^s + 2\gamma( L(y^s, \tau^s) - d(\tau^{s+1})  - 2\epsilon^s )\\
   \leq  &2\gamma \thickhat{\nabla}d(\tau^s)^\top(\tau^{s}-\tau) + (\textstyle\frac{2\norm{A}_2^2}{\alpha}\gamma - 1) \norm{\tau^{s+1} - \tau^s}_2^2,
\end{aligned}
\end{equation}
and use the fact that $\gamma\textstyle\frac{2\norm{A}_2^2}{\alpha} \leq 1$ to eliminate the $\norm{\tau^{s+1} - \tau^s}_2^2$ term and complete the proof.
\end{pf}
\subsubsection{Proof of Proposition~\ref{prop:dual_properties}}\label{app:dual_property_proof}
$d(\tau)$ is the dual function of the optimization problem
\[\miny F(x) \text{ s.t. } Ay \leq b.\]
As the dual function of a convex optimization with linear constraints, it is concave~\cite[Prop 5.1.2]{bertsekas1999nonlinear}.  The smoothness constant of $d(\tau)$ follows from~\cite[Thm 1]{nesterov2005smooth},
where $\alpha$ is the strong convexity factor of $F_0$. Finally, the computation of $\nabla d(\tau)$ follows directly from~\cite[Prop.B.25]{bertsekas1999nonlinear}.
\subsection{Proof of Theorem~\ref{thm:primal_convergence}}
Let $r^s= \norm{\tau^s - \tau^\star}_2^2$. From Lemma~\ref{lem:residue_inequality}, \change{when $\gamma \leq \frac{\alpha}{2\norm{A}^2_2}$}, we have
\begin{align} \label{eqn: thm1 eqn1}
 &\textstyle r^{s+1} \leq   \\
 &\textstyle r^s + 2\gamma\big(d(\tau^{s+1}) - L(y^s, \tau^s)  + 2\epsilon^s + \thickhat{\nabla} d(\tau^s)^\top(\tau^{s} - \tau^\star)\big) \label{eqn:proof_convergence_3}\nonumber
\end{align}
From Lemma~\ref{lem:approx_concavity}, we have
\begin{equation}\label{eqn: thm1 eqn3}
    \textstyle \thickhat{\nabla} d(\tau^s)^\top (\tau^s-\tau^\star)\leq L(y^s, \tau^s)-d(\tau^\star)
\end{equation}
Summing up  \eqref{eqn: thm1 eqn1} and \(2\gamma\times\)\eqref{eqn: thm1 eqn3}, we obtain
\begin{equation}\label{eqn: thm1 eqn4}
    r^{s+1}-r^s\leq  2\gamma (d(\tau^{s+1})-d(\tau^\star)+2\epsilon^s)
\end{equation}
Summing over \eqref{eqn: thm1 eqn4} for $s= 0\ldots, k-1$, we obtain $0 \leq r^k \leq r^0 -2\gamma \sum_{s=1}^{k} \big(d(\tau^\star) - d(\tau^{s})\big) +  4\gamma\sum_{s=0}^{k-1} \epsilon^s$. 
Finally, the concavity of $d$ from Proposition~\ref{prop:dual_properties} implies that  
$-kd(\thickbar{\tau}^k) = -kd(\sum_{s=1}^k \tau^s)\leq -\sum_{s=1}^kd(\tau^s)$. This completes the proof.
\subsection{Proof of Theorem~\ref{thm:dualAscent}}\label{app:a3}
We bound the term $F(\thickbar{y}^k) - F(y^\star)$. First consider the upper bound.
From Lemma~\ref{lem:residue_inequality}, let $\tau = 0$,
\begin{equation}\label{eqn:proof_y_convergence_1}
\begin{aligned}
    \textstyle  \norm{\tau^{s+1}}_2^2 \leq \norm{\tau^s}_2^2 +2\gamma \big(&d(\tau^{s+1}) + 2\epsilon^s - L(y^s, \tau^s) \\ 
      &\textstyle + \thickhat{\nabla}d(\tau^s)^\top\tau^s\big).  
\end{aligned}
\end{equation}
Recall from~\eqref{eqn:dual_gradient} and~\eqref{eqn:constrained_lagrangian}, $\thickhat{\nabla}d(\tau^s) = Ay^s - b$ and $L(y^s, \tau^s) = F(y^s) + (\tau^s)^\top(Ay^s - b)$. Therefore $L(y^s, \tau^s) - \thickhat{\nabla}d(\tau^s)^\top\tau^s = F(y^s)$. Then~\eqref{eqn:proof_y_convergence_1} becomes
\begin{equation}
\norm{\tau^{s+1}}_2^2 + 2\gamma(F(y^s) - d(\tau^{s+1})) \leq \norm{\tau^s}_2^2 + 4\gamma \epsilon^s
\end{equation}
Summing over $s = 0,\ldots k-1$, $\sum_{s=0}^{k-1} F(y^s)- d(\tau^{s+1}) \leq \frac{1}{2\gamma}\norm{\tau^0}_2^2 + 2E^k$.
Taking the average $\thickbar{y}^k$ and noting that $d(\tau^k) \leq d(\tau^\star) = F(y^\star)$ for all $\tau^k \in \reals^C_+$,
\begin{equation}\label{eqn:proof_primal_lowerbound}
    \textstyle F(\thickbar{y}^k)- F(y^\star) \leq \frac{1}{2\gamma k}\norm{\tau^0}_2^2 +\frac{2E^k}{k}.
\end{equation}
Next, consider the lower bound of $F(\thickbar{y}^k) - F(y^\star)$. By definition, $y^\star$ solves $\miny F(y) + (Ay - b)^\top\tau^\star$ where $(Ay^\star - b)^\top\tau^\star = 0$. This implies that $F(y^\star) \leq L(\thickbar{y}^k, \tau^\star)$. We expand $L(\thickbar{y}^k, \tau^\star)$ with~\eqref{eqn:constrained_lagrangian} to obtain
\[\textstyle F(y^\star) - F(\thickbar{y}^k) \leq  (A\thickbar{y}^k - b)^\top \tau^\star \leq [A\thickbar{y}^k - b]_+^\top\tau^\star.\]
We can then bound  the difference $\textstyle F(y^\star) - F(\thickbar{y}^k)$ by $\norm{\tau^\star}_2\norm{[A\thickbar{y}^k - b]_+}_2$. From Corollary~\ref{cor:constraint_violation_convergence},
\begin{equation} \label{eqn:proof_primal_upperbound}
    \textstyle F(y^\star) - F(\thickbar{y}^k) \leq \frac{\norm{\tau^\star}_2}{\gamma k}(\norm{\tau^\star}_2 + \norm{\tau^0 - \tau^\star}_2 + 2\sqrt{\gamma E^k}).
\end{equation} 
Together, \eqref{eqn:proof_primal_lowerbound} and \eqref{eqn:proof_primal_upperbound} imply
\begin{equation}\label{eqn:proof_primal_finalbound}
    \textstyle | F(y^\star) - F(\thickbar{y}^k)| \leq \frac{1}{\gamma k}D(\tau^\star, \tau^0, E^k)
\end{equation}
Strong convexity of $F$ follows from Assumption~\ref{ass:strong_convexity}, such that  $\norm{\thickbar{y}^k - y^\star}_2^2 \leq \frac{\alpha}{2}|F(\thickbar{y}^k )- F(y^\star)|$. This combined with~\eqref{eqn:proof_primal_finalbound} completes the proof.

\subsection{Proof of Corollary~\ref{cor:constraint_violation_convergence}}
We first derive an upper bound for $\norm{\tau^{k}}_2$ and then bound the left hand side of~\eqref{eqn:constraint_violation_convergence} by $\norm{\tau^{k}}_2$. Recall~\eqref{eqn: thm1 eqn4}, we use $d(\tau^\star) - d(\tau^{k})\geq 0$ to derive $r^{s+1} \leq r^s + 4\gamma \epsilon^s$. Summing over $s = 0,\ldots, k-1$, we have
\begin{equation}\label{eqn: cor1 eqn1}
   \textstyle  \norm{\tau^{k} - \tau^\star}_2^2 \leq \norm{\tau^{0} - \tau^\star}_2^2 + \textstyle4\gamma E^k.
\end{equation}
Taking the square root of both sides of~\eqref{eqn: cor1 eqn1} and noting the identity $\sqrt{a + b} \leq \sqrt{a} + \sqrt{b}$, we obtain
\begin{equation}\label{eqn: cor1 eqn2}
    \norm{\tau^{k} - \tau^\star}_2 \leq \norm{\tau^{0} - \tau^\star}_2 + \textstyle\sqrt{4\gamma E^k}.
\end{equation}
We add $\norm{\tau^\star}_2$ to both sides of~\eqref{eqn: cor1 eqn2} and use the triangle inequality $\norm{\tau^k}_2 \leq \norm{\tau^{k} - \tau^\star}_2 + \norm{\tau^\star}_2 $ to obtain
\begin{equation}\label{eqn: cor1 eqn3}
    \norm{\tau^{k}}_2 \leq \norm{\tau^\star}_2 + \norm{\tau^{0} - \tau^\star}_2 + \textstyle\sqrt{4\gamma E^k}.
\end{equation}
Next, we bound $\norm{[A\thickbar{y}^k - b]_+}_2$ using $\norm{\tau^{k}}_2$. From line~\ref{algline:dual_ascent_tau_update} of Algorithm~\ref{alg:dualAscent}, $\tau^{s+1} \geq \tau^s + \gamma(Ay^s - b)$. We sum over $s = 0,\ldots, k-1$ to obtain  $\tau^{k} \geq \tau^0 + \gamma k (A\thickbar{y}^k - b)$. Noting $\tau^{0} \in \reals_+^C$ can be dropped, $\gamma k [A\thickbar{y}^k - b]_+  \leq  \tau^{k}$ combined with~\eqref{eqn: cor1 eqn3} completes the proof.

\subsection{Proof of Proposition~\ref{prop:dual_properties}}\label{sec:dual_function_property}
\begin{pf}
$d(\tau)$ is the dual function of the optimization problem
\[\miny F(x) \text{ s.t. } Ay \leq b.\]
As the dual function of a convex optimization with linear constraints, it is concave~\cite[Prop 5.1.2]{bertsekas1999nonlinear}.  The smoothness constant of $d(\tau)$ follows from~\cite[Thm 1]{nesterov2005smooth},
where $\alpha$ is the strong convexity factor of $F_0$. Finally, the computation of $\nabla d(\tau)$ follows directly from~\cite[Prop.B.25]{bertsekas1999nonlinear}.
\end{pf}
\subsection{Proof of Lemma~\ref{lem:approx_concavity}}\label{app:lemma_approx_concavity_proof}
\begin{pf}
We denote $\hat{y}_{\tau}(\epsilon)$ by $\hat{y}_\tau$ for simplicity.
Since \(\hat{y}_{\tau}\in\mc{W}(\ell_\tau, \epsilon)\subset\mc{Y}(P, p_0)\), using \eqref{eqn:dual_function_and_equilibria} we can show that
\begin{equation}\label{eqn:lem1 eq 1}
    d(\sigma)\leq L(\hat{y}_{\tau}, \sigma).
\end{equation}
Combining~\eqref{eqn:lem1 eq 1} with the fact that \(L(\hat{y}_{\tau}, \sigma)=L(\hat{y}_{\tau}, \tau)+\thickhat{\nabla}d(\tau)^\top (\sigma-\tau)\), we obtain Lemma~\ref{lem:approx_concavity}.
\end{pf}

\subsection{Proof of Lemma~\ref{lem:approx_smoothness}}\label{app:lemma_approx_smoothness_proof}
\begin{pf}
We denote $\hat{y}_{\tau}(\epsilon)$ by $\hat{y}_\tau$ for simplicity and recall $y_\tau$ from~\eqref{eqn:dual_function_and_equilibria}. From Proposition~\ref{prop:dual_properties}, we know that
\begin{equation}\label{eqn: lem1 eqn1}
    \nabla d(\tau)=\thickhat{\nabla}d(\tau)+A(y_\tau-\hat{y}_\tau).
\end{equation}
Substituting \eqref{eqn: lem1 eqn1} into \eqref{eqn:d_smoothness}, we obtain the following
\begin{equation}\label{eqn: lem1 eqn2}
   \begin{aligned}
   0 \leq &  d(\sigma) - d(\tau) - \thickhat{\nabla}d(\tau)^\top(\sigma - \tau) \frac{\norm{A}_2^2}{2\alpha}\norm{\sigma - \tau}_2^2\\
    &- \big(A(y_\tau - \hat{y}_\tau)\big)^\top(\sigma - \tau). 
\end{aligned} 
\end{equation}
Furthermore, we can show
\begin{equation}\label{eqn: lem1 eqn3}
   \begin{aligned}
    \Big|\big(A(y_\tau - \hat{y}_\tau)\big)^\top(\sigma - \tau)\Big| \leq &\norm{\hat{y}_\tau - y_\tau}_2\cdot\norm{A}_2 \cdot\norm{\sigma - \tau}_2 \\
    \leq & \frac{\alpha}{2}\norm{\hat{y}_\tau - y_\tau}_2^2 + \frac{\norm{A}_2^2}{2{\alpha}}\norm{\sigma - \tau}_2^2,
\end{aligned} 
\end{equation}
where the first step is due to the Cauchy–Schwarz inequality, and the second step is due to the inequality of arithmetic and geometric inequalities.

Next, we note that \ctwo{$F$~\eqref{eqn:potential_def}} and subsequently $L(y, \tau)$~\eqref{eqn:constrained_lagrangian} is strongly convex under Assumption~\ref{ass:strong_convexity}. We combine this with the fact that $L(y_\tau, \tau)=d(\tau)$ from~\eqref{eqn:dual_function_and_equilibria} to obtain 
\begin{equation}\label{eqn: lem1 eqn4}
    \frac{\alpha}{2}\norm{\hat{y}_\tau - y_\tau}_2^2\leq L(\hat{y}_\tau, \tau)-d(\tau).
\end{equation}
From~\eqref{eqn:eps_tau_wardrop_eq}, $\hat{y}_\tau$ satisfies
\ctwo{
\begin{equation}\label{eqn: lem1 eqn5}
    L(\hat{y}_\tau, \tau) - d(\tau) \leq \epsilon.
\end{equation}}
Summing up\ctwo{~\eqref{eqn:eps_tau_wardrop_eq}},~\eqref{eqn: lem1 eqn2}, \eqref{eqn: lem1 eqn3}, \eqref{eqn: lem1 eqn4}, and $2\times$\eqref{eqn: lem1 eqn5}, we obtain~\eqref{eqn:pseudo_smoothness}, which completes the proof. 
\end{pf}

\subsection{Lemma~\ref{lem:residue_inequality}}
\begin{lem}\label{lem:residue_inequality}
Under Assumption~\ref{ass:strong_convexity}, if \ctwo{$\gamma\leq \frac{\alpha}{2\norm{A}_2^2}$}, $\tau^s$ from Algorithm~\ref{alg:dualAscent} satisfy
\begin{equation}\label{eqn:residue_inequality}
\begin{aligned}
    \norm{\tau^{s+1} - \tau}_2^2 \leq &\norm{\tau^s - \tau}_2^2 + 2\gamma\Big(d(\tau^{s+1}) - L({y}^s, \tau^s) + 2\epsilon^s \\
   & + \thickhat{\nabla}d(\tau^s)^\top(\tau^s - \tau)\Big), \quad \forall \ \tau \in \reals^C_+, \ k \geq 0.
\end{aligned}
\end{equation}
\end{lem}
\begin{pf}
Given $\tau \in \reals_+^C$, let $r^s = \norm{\tau^s - \tau}_2^2$. We compute $r^{s+1} - r^s$ using the law of cosine as  
\begin{equation}\label{eqn: lem3 eqn1}
\begin{aligned}
 r^{s+1}- r^s = &  2(\tau^{s+1} - \tau^s)^\top(\tau^{s+1} - \tau) - \norm{\tau^{s+1} - \tau^s}_2^2.
\end{aligned}
\end{equation}
From line~\ref{algline:dual_ascent_tau_update} of Algorithm~\ref{alg:dualAscent},   \(\tau^{s+1} = [\tau^s+\gamma(Ay^s-b)]_+\). Using~\cite[Lem 3.1]{bubeck2015convex}, the projection  onto \(\mathbb{R}_+^C\) implies that
\begin{equation}\label{eqn: lem3 eqn2}
    0\leq (\tau^s+\gamma\thickhat{\nabla }d(\tau^s)-\tau^{s+1})^\top(\tau^{s+1}-\tau)
\end{equation}
From~\eqref{eqn: lem3 eqn2}, we can upper bound $(\tau^{s+1} - \tau^s)^\top(\tau^{s+1} - \tau) $ and combine with~\eqref{eqn: lem3 eqn1} to obtain
\begin{equation}\label{eqn: lem 3 eqn 3}
\begin{aligned}
& \textstyle r^{s+1} - r^s  \leq \textstyle  2\gamma\thickhat{\nabla} d(\tau^s)^\top(\tau^{s+1} - \tau) - \norm{\tau^{s+1} - \tau^s}_2^2
\end{aligned}
\end{equation}
From Lemma~\ref{lem:approx_smoothness}, we recall
\begin{equation}\label{eqn: lem 3 eqn 4}
\begin{aligned}
  & L(y^s, \tau^s) - d(\tau^{s+1})  - 2\epsilon^s \\
  \leq  &\thickhat{\nabla}d(\tau^s)^\top(\tau^{s}-\tau^{s+1})  + \frac{\norm{A}_2^2}{\alpha}\norm{\tau^{s+1} - \tau^s}_2^2 .
\end{aligned}
\end{equation}
We can then combine~\eqref{eqn: lem 3 eqn 3} and $2\gamma\times$\eqref{eqn: lem 3 eqn 4} to derive 
\begin{equation}\label{eqn: lem 3 eqn 5}
\begin{aligned}
   &r^{s+1} - r^s + 2\gamma( L(y^s, \tau^s) - d(\tau^{s+1})  - 2\epsilon^s )\\
   \leq  &2\gamma \thickhat{\nabla}d(\tau^s)^\top(\tau^{s}-\tau) + (2\gamma \frac{\norm{A}_2^2}{\alpha} - 1) \norm{\tau^{s+1} - \tau^s}_2^2,
\end{aligned}
\end{equation}
and use the fact that $2\gamma\frac{\norm{A}_2^2}{\alpha}\leq 1$ to eliminate the $\norm{\tau^{s+1} - \tau^s}_2^2$ term and complete the proof.
\end{pf}

\end{document}